\documentclass{amsart}
\usepackage[english]{babel}
\usepackage[T1]{fontenc}
\usepackage[utf8]{inputenc}
\usepackage{paralist}
\usepackage{enumitem}
\usepackage{caption}
\usepackage{subcaption}

\usepackage{amsmath,amsfonts,amssymb}
\usepackage{upgreek}
\usepackage{nicefrac}
\usepackage{faktor}
\usepackage{xspace}

\usepackage[ruled,linesnumbered]{algorithm2e}
\makeatletter
\newcommand{\algorithmstyle}[1]{\renewcommand{\algocf@style}{#1}}
\makeatother

\SetKwInOut{Parameter}{Parameters}
\SetKw{And}{and}
\SetKw{Goto}{goto}
\SetKw{True}{true}
\SetKw{Break}{break}
\SetKw{False}{false}
\SetKw{Downto}{downto}
\SetKw{To}{to}
\usepackage{xcolor}
\definecolor{Eggplant}{RGB}{97, 64, 81}
\definecolor{RoyalBlue}{cmyk}{1, 0.50, 0, 0}
\usepackage{bbm}
\usepackage{multirow}
\usepackage{graphicx}
\usepackage{tabularx}
\usepackage{booktabs}

\usepackage{tikz}
\usetikzlibrary{matrix}

\usepackage{todonotes}

\usepackage{amsthm}
\usepackage{mathtools}

\let\originalleft\left
\let\originalright\right
\renewcommand{\left}{\mathopen{}\mathclose\bgroup\originalleft}
\renewcommand{\right}{\aftergroup\egroup\originalright}

\newcommand{\GSO}{\mbox{\textsc{gso}\xspace}}
\newcommand{\LL}{\mbox{\textsc{l}$^2$}\xspace}
\newcommand{\LLL}{\mbox{\textsc{lll}}}
\newcommand{\TLL}{\mbox{$\overline{\textrm{\textsc{ll}}}$}}
\newcommand{\ALLL}{\mbox{\textsc{adaptive-lll}}\xspace}

\DeclareMathOperator{\rk}{rk}

\DeclareMathOperator{\covol}{covol}

\newcommand{\vect}[1]{{#1}}
\newcommand{\inter}[1]{\underline{#1}}
\newcommand{\Vol}[1]{\covol\left({#1}\right)}
\newcommand{\inner}[2]{\langle \vect{#1},\vect{#2} \rangle}
\newcommand{\Gram}{\mathcal{G}}
\newcommand{\flag}[1]{\mathcal{F}_{#1}}

\newcommand{\NN}{\mathbf{N}}
\newcommand{\ZZ}{\mathbf{Z}}
\newcommand{\RR}{\mathbf{R}}

\newcommand{\KK}{\mathbf{K}}
\newcommand{\QQ}{\mathbf{Q}}
\newcommand{\CC}{\mathbf{C}}

\newcommand{\Gl}{\textrm{Gl}}
\newcommand{\Sym}{\mathcal{S}}

\newcommand{\Lat}{\Lambda}
\newcommand{\Bas}{\mathcal{B}}

\newcommand{\order}{\mathfrak{o}}

\newcommand{\Mult}[1]{\mathcal{M}\left({#1}\right)}

\newcommand{\bigO}[1]{\textrm O\left(#1\right)}
\newcommand{\littleO}[1]{o\left(#1\right)}

 \usepackage{amsmath}
  \usepackage{paralist}
  \usepackage{graphics}
  \usepackage{epsfig}
\usepackage{graphicx}  \usepackage{epstopdf}
 \usepackage[colorlinks=true]{hyperref}
\hypersetup{urlcolor=RoyalBlue, citecolor=Eggplant, linkcolor=Eggplant}

\newtheorem{theorem}{Theorem}[section]

\newtheorem{lemma}[theorem]{Lemma}
\newtheorem{proposition}{Proposition}

\newtheorem*{problem}{Problem}
\theoremstyle{definition}
\newtheorem{definition}[theorem]{Definition}

\title[Certified Lattice reduction]
      {Certified Lattice reduction}

\author[Thomas Espitau and Antoine Joux]{}


 \email{thomas.espitau@lip6.fr}
 \email{antoin.joux@m4x.org}

\thanks{This work has been supported in part by the European Union as H2020
Programme under grant agreement number ERC-669891.}

\begin{document}
\maketitle

\centerline{\scshape Thomas Espitau}
\medskip
{\footnotesize
 \centerline{Sorbonne Universit\'e}
   \centerline{LIP 6, CNRS UMR 7606}
   \centerline{Paris, France}
}

\medskip

\centerline{\scshape Antoine Joux}
\medskip
{\footnotesize
  \centerline{\it Chaire de Cryptologie de la Fondation SU}
  \centerline{Sorbonne Universit\'e, Institut de Math\'ematiques de
  Jussieu--Paris Rive Gauche}
  \centerline{CNRS, INRIA, Univ Paris Diderot.}
  \centerline{Campus Pierre et Marie Curie, F-75005, Paris, France}
}

\begin{abstract}
Quadratic form reduction and lattice reduction are fundamental tools
  in computational number theory and in computer science, especially
  in cryptography.  The celebrated Lenstra--Lenstra--Lovász reduction
  algorithm (so-called \LLL) has been improved in many ways through
  the past decades and remains one of the central methods used for
  reducing \emph{integral} lattice basis. In particular, its
  floating-point variants---where the rational arithmetic required by
  Gram--Schmidt orthogonalization is replaced by floating-point
  arithmetic---are now the fastest known. However, the systematic
  study of the reduction theory of \emph{real} quadratic forms or, more
  generally, of real lattices is not widely represented in the
  literature. When the problem arises, the lattice is usually replaced
  by an integral approximation of  (a multiple of) the original
  lattice, which is then reduced. While practically useful and proven
  in some special cases, this method doesn't offer any guarantee of
  success in general.  In this work, we present an adaptive-precision
  version of a generalized \LLL~algorithm that covers this case in all
  generality. In particular, we replace floating-point arithmetic by
  Interval Arithmetic to certify the behavior of the algorithm. We
  conclude by giving a typical application of the result in algebraic
  number theory for the reduction of ideal lattices in number fields.
\end{abstract}

\section{Introduction}
In a general setting, a \emph{lattice} $\Lat$ is a free $\ZZ$-module
of finite rank, endowed with a positive-definite bilinear form on its
ambient space $\Lat\otimes_\ZZ \RR$, as presented for instance
in~\cite{LenstraS17}. In particular, this definition implies that
$\Lat$ is discrete in its ambient space for the topology induced by
the scalar product. This formalism encompasses the well-known
\emph{Euclidean lattices} when taking the canonical scalar product of
$\RR^d$, but also lattices arising from ideals in rings of integers of
number fields.
The rank of the lattice $\Lat $ is defined as the dimension
of the vector space $\Lat\otimes_\ZZ \RR$.  By definition of a
finitely-generated free module, there exists a finite set of vectors
$b_1, \ldots, b_{\rk \Lat} \in \Lat$ such that
$\Lat = \bigoplus_{i=1}^{\rk \Lat} b_i \ZZ$. Such a family is called a
basis of the lattice and is not unique. In fact, as soon as
$\rk \Lat \geq 2$ there are infinitely many bases of $\Lat$. Some
among those have interesting properties, such as having reasonably
small vectors and low orthogonality defects. They are informally
called \emph{reduced bases} and finding them is the goal of
\emph{lattice reduction}.

Numerous algorithms
arising in algebraic number theory heavily rely on lattice reduction,
for example, the computation of normal forms of integral matrices
(see~\cite{Jager05} for the Hermite Normal Form and~\cite{Havas98} for
the Smith Normal Form), class group computations in a number
field~\cite{GelinJ16,BiasseF13},
or even the enumeration of points of
small height near algebraic curves~\cite{Elkies00}.

Even for lattices that use the canonical scalar product, there is a
deep link with bilinear forms that clearly appears when considering
the \emph{Gram matrix} of a basis
$\Bas=\{\vect{b_1}, \dotsc, \vect{b_d}\}$, that is, the real symmetric
matrix $\Gram = \left( \inner{b_i}{b_j}\right)_{i,j}$.
\medskip

The study of these reduction problems is not recent and goes back to
the works of Lagrange and Gauss. These early works were expressed in
terms of reduction of quadratic form, more precisely integral binary
quadratic forms\footnote{This can be viewed as the reduction of
  integral dimension-two lattices.} and led to a
method often called Gauss' algorithm. This method can be seen as a
2-dimensional extension of the Euclid algorithm for computing the
greatest common divisor of two integers.  In 1850, Hermite proved a
general upper bound on the length of the shortest vector in a lattice,
given as a function of the dimension and of a very important invariant
called the determinant, which is defined in
Section~\ref{sec:gram_schmidt}. This bound involves the so-called
Hermite constant and has recently been rephrased in algorithmic
terms~\cite[Hermite's Algorithms]{Nguyen2010}.  A century later, in
1982, Lenstra, Lenstra and Lovász designed the
\LLL~\emph{algorithm}~\cite{LLL82}, with the polynomial factorization
problem as an application, following the work of Lenstra on integer
programming~\cite{Lenstra1983}.  This algorithm constitutes a
breakthrough in the history of lattice reduction algorithm, since it
is the first to have a runtime polynomial in terms of the
dimension. It was followed by many improvements lowering
its complexity or improving the output's quality.  \medskip

Current implementations of \LLL{} often work with low precision
approximations in order to greatly speed-up the computations. Indeed,
the algorithm works surprinsingly well even with such reduced
precisions, even if  some care needs to be taken to avoid infinite
loops. Moreover, once the result is obtained, it can verified
efficiently as shown in~\cite{Villard07}.\medskip

We propose here an alternative strategy where we not only certify that
the end-result is a reduced basis but also that the algorithm followed
a valid computation path to reach it. This strongly deviates from
other approaches that have been taken to obtain guaranteed lattice
reduced basis. At first, this may seems irrelevant. After all, one
might claim that a basis satisfying the end conditions of \LLL{} is what
is desired and that the computation path doesn't matter. However, as
shown in~\cite{KimVenkatesh16} for Siegel-reduced bases, a reduced
basis chosen uniformely at random behaves as the worst-case allowed by
the final inequalities. By constrast, bases produced by the \LLL{}
algorithm are usually much better than this worst-case. This argues in
favor of trying to follow the algorithm defintion exactly to better
understand the phenomenon. In particular, this option might be
invaluable for experiments performed toward analyzing this gap.\medskip

The present article also relies on Interval Arithmetic, a
representation of reals by intervals---whose endpoints are
floating-point numbers---that contain them.  Arithmetic operations, in
particular the basic operations $+, -, \times , \div$ can be redefined
in this context.  The main interest of this representation lies in its
\emph{certification} property: if real numbers are represented by
intervals, the interval resulting from the evaluation of an algebraic
expression contains the exact value of the evaluated
expression.\medskip

For some authors, Interval Arithmetic was introduced by
R. Moore in 1962 in his Ph.D. thesis~\cite{Moore62}. For others, it can
be dated back to 1958, in an article of T. Sunaga~\cite{Sunaga58} which
describes an algebraic interpretation of the lattice of real intervals,
or even sooner in 1931 as a proposal in the Ph.D. thesis~\cite{Young31}
of R.C.~Young at Cambridge.
Its main
asset---calculating
directly on sets---is nowadays used to deterministically determine
the global extrema of a continuous function~\cite{Ratschek} or
localizing the zeroes of a function and (dis)proving their
existence~\cite{jaulin2001}. Another application of Interval Arithmetic
is to be able to detect lack of precision at run-time of numerical
algorithms, thanks to the guarantees it provides on computations.
This can, in particular, be used to design adaptive-precision
numerical algorithms.
 \medskip

In the present paper, we propose to transform and generalize
the \LLL~algorithm into an adaptive-precision version, which can reduce
arbitrary lattices and follows a certified flow of execution. More
precisely, it uses Interval Arithmetic to validate the size-reduction
and exchange steps that occur within \LLL.

The interested reader may download an implementation of the algorithm
from the webpage \url{http://almacrypt.eu/outputs.php}.

\subsection*{Organisation of the paper}
In Section 2, we briefly introduce reduction theory and present the \LL
variant of the \LLL~algorithm.  Section 3 aims at describing the basics
of Interval Arithmetic used in Section 4 to handle the problem of
representation of real lattices.  The framework of this latter section
is then used in Section 5 to derive a certified reduction algorithm for
real lattices. Section 6 presents an application to
algorithmic number theory.

\subsection*{Notations and conventions}

\subsubsection*{General notations}
As usual, the bold capitals $\ZZ$, $\QQ$, $\RR$ and $\CC$ refer respectively to the
ring of integers and the fields of rational, real and
complex numbers.  Given a real number $x$, the integral roundings
\emph{floor}, \emph{ceil} and \emph{round to nearest integer} are
denoted respectively by
$\lfloor x\rfloor, \lceil x\rceil, \lfloor x \rceil$. Note that the
rounding operator is ambiguous when operating on
half-integers. However, either choice when rounding is acceptable in
lattice reduction algorithms. In fact, in this context, it is often enough to return an
integer close to $x$, not necessarily the closest.

These operators are extended to operate on vectors and matrices by point-wise
composition. The complex conjugation of $z\in\CC$ is denoted by the
usual bar $\bar{z}$ whereas the real and imaginary parts of a complex
$z$ are indicated by respectively $\mathfrak{R}(z)$ and
$\mathfrak{I}(z)$. All logarithms are taken in base $2$.

\subsubsection*{Matrices and norms}
For a field $\KK$, let us denote by $\KK^{d\times d}$ the space of
square matrices of dimension $d$ over $\KK$, $\Gl_d(\KK)$ its group of
invertible elements and $\Sym_d(\KK)$ its subspace of symmetric matrices.  For a
complex matrix $A$, we write $A^\dagger$ for its conjugate transpose.
For a vector $v$, we denote by $\|v\|_\infty$ its absolute (or infinity)
norm, that is the maximum of the absolute value of its entries. We
similarly define the matrix \emph{max norm} $\| B \|_\textrm{max} =
\max_{(i,j)\in [1\,\cdots\,d]^2} |B_{i,j}|$, for any matrix~$B$.

\subsubsection*{Computational setting}
The generic complexity model used in this work is the random-access
machine  (RAM) model and the computational cost is measured in bits
operations.  $\Mult{k}$ denotes the complexity of the multiplication of
two integers of bit length at most $k$. It is also the cost of the
multiplication of two floating-point numbers at precision $k$, since
the cost of arithmetic over the exponents is negligible with regards to
the cost of arithmetic over the mantissae.

\section{Basics of Lattice Reduction}
\subsection{Orthogonalization}
\label{sec:gram_schmidt}
Let us fix an Euclidean space $(E,\inner{\cdot}{\cdot})$, i.e. a
real vector space $E$ together with a positive-definite bilinear form
$\inner{\cdot}{\cdot}: E\times E \to \RR$. As usual, two vectors
$x,y\in E$ are said to be orthogonal---with respect to the form
$\inner{\cdot}{\cdot}$---if $\inner{x}{y} = 0$. More generally a
family of vectors is orthogonal if its elements are pairwise
orthogonal.

Now consider $S = (b_1,\dots, b_r),$ a family of linearly independent
vectors of $E$. The \emph{flag} $\flag{S}$ associated to $S$ is the
finite increasing chain of subspaces:
\[
  b_1 \RR \subset b_1 \RR  \oplus b_2\RR \subset \cdots \subset
  \bigoplus_{i=1}^r b_i\RR.
\]
The orthogonal complement $S^\bot$ is defined as the subspace
$\{ x\in E~|~ \forall i,~\inner{x}{b_i} = 0\}$. Denote by $\pi_i$ the
orthogonal projection on $(b_{1}, \dotsc, b_{i-1})^\bot$, with the
convention that $\pi_1$ is the identity map. The Gram--Schmidt orthogonalization
process---shorthanded as \GSO---is an algorithmic method for
orthogonalizing $S$ while preserving its flag. It constructs
the orthogonal set
$S^* = \left(\pi_1(b_1), \ldots, \pi_r(b_r)\right)$. The computation
$S^*$ can be done inductively as follows:
\[
  \begin{aligned}
    \pi_1(b_1) &= b_1\\
    \forall 1<i\leq r,\quad  \pi_i(b_i) &= b_i - \sum_{j=1}^{i-1}
    \frac{\inner{b_i}{\pi_j(b_j)}}{\inner{\pi_j(b_j)}{\pi_j(b_j)}}b_j.
\end{aligned}
\]
Define the \emph{Gram matrix}, associated to a family of vectors
$S = (b_1, \ldots, b_r)$, as the symmetric matrix of scalar products:
$\Gram_S = \left(\inner{b_i}{b_j}\right)_{(i,j)\in [1\,\cdots\,
  r]^2}$. The (co)volume of~$S$, also called its determinant, is
defined as the square root of the Gram determinant $\det \Gram_S$.  It
can be easily computed from the Gram-Schmidt vectors $S^*$ as:
\[
\Vol{S} = \prod_{i=1}^r \|\pi_i(b_i)\| \]

\subsection{Lattices and reduction}

\begin{definition}
  A (real) \emph{lattice}  $\Lat$ is a finitely generated free
  $\ZZ$-module,
  endowed with a positive-definite bilinear form $\inner{\cdot}{\cdot}$
  on its ambient space $\Lat\otimes_\ZZ \RR$.

  By definition of the tensor product, there is a canonical injection
  that sends a vector $v$ to $v\otimes 1$ in the ambient space and
  preserves linear independence.  Thus, the rank of $\Lat$ as a
  $\ZZ$-module, is equal to the dimension of the vector space
  $\Lat\otimes_\ZZ \RR$.

  Denoting by $d$ the rank of the lattice, a \emph{basis} of $\Lat$ is
  a family $b_1, \ldots, b_d$ of elements of $\Lat$ such that
  $\Lat = \bigoplus_{i=1}^d b_i \ZZ$.
\end{definition}
In the sequel, we identify $\Lat$ with its canonical image
$\Lat\otimes 1$ and thus view the lattice as an additive subgroup of
its ambient space $\Lat\otimes_\ZZ \RR$.  When the context makes it
clear, we may omit to write down the bilinear form associated to a
lattice $\Lat$.  Throughout this section, $\|\cdot\|$ stands for the
Euclidean norm induced by $\inner{\cdot}{\cdot}$, unless stated
otherwise.  As usual, any two bases $(b_1, \ldots, b_d)$ and
$(b'_1, \ldots, b'_d)$ of $\Lat$ are related by a unimodular
transformation, i.e., a linear transformation represented by a
$d\times d$ integer matrix of determinant $\pm 1$.

\begin{lemma}
A lattice $\Lat$ is discrete for the topology induced by the given
norm on its ambient space. I.e., there exists a real $\epsilon_\Lat>0$ such
that for any pair $(x,y)$ of elements of $\Lat$ with $x\neq y$ we
have:
$$
\| x-y \| \geq \epsilon_\Lat.
$$

The largest possible value for $\epsilon_\Lat$ in the above inequality
is equal to the norm of the shortest non-zero vector of $\Lat$, which
is traditionally called the \emph{first minimum} or the \emph{minimum
  distance} of the lattice and denoted by $\lambda_1 (\Lat)$.
\end{lemma}
\begin{proof}
  Let $\Bas = (b_1, \ldots, b_d)$ be a basis of $\Lat$. Let
  $\Bas^* = (\pi_1(b_1), \ldots, \pi_d(b_d))$ be the orthogonal basis obtained by
  applying Gram-Schmidt orthogonalization to the canonical image of
  $\Bas$ in $\Lat\otimes_\ZZ \RR$. This orthogonalization is taken using
  as scalar product the given bilinear form.

Assume by contradiction that there exist pairs of distinct vectors
with the norm of their difference arbitrarily small.  Since the
difference is also an element of $\Lat$, there are non-zero elements
of arbitrarily small norm in $\Lat$.
For any integer $i>0$, choose a
vector $x_i$ in $\Lat$ with $\|x_i\|^2\leq 2^{-i}.$ Decompose
$x_i$ in the basis $\Bas^*$ as
$
x_i=\sum_{j=1}^{d}\chi_i^{(j)}\pi_j(b_j).
$
For any pair of integers $i$, $j$ we see that
$
|\chi_i^{(j)}|^2\,\|\pi_j(b_j)\|^2\leq \|x_i\|^2\leq 2^{-i}.
$
As a consequence, each sequence $\chi^{(j)}$ converges to
zero. Multiplying by the basis-change matrix, we see that the coordinates
of $x_i$ in the basis $b_1, \ldots, b_d$ also converge to zero. Since these
coordinates are integral, the sequences are ultimately constant and
$x$ is also ultimately constant (and null). This contradicts the choice of
$x_i$ as a non-zero element.
\end{proof}
\medskip

\subsection{The LLL~reduction algorithm}
In 1982, Lenstra, Lenstra and Lovász~\cite{LLL82} proposed a
notion called \emph{\LLL~reduction} and a polynomial-time
algorithm that computes an {\LLL}-reduced basis from arbitrary basis
of the same lattice. Their reduction notion is formally defined as follows:

\begin{definition}[\LLL~reduction]
\label{LLLstdDef}
  A basis   $\Bas = (b_1, \ldots, b_d)$ of a lattice is said to be
  $\delta$-\LLL-\textbf{reduced} for  a parameter $1/4<\delta<
  1$, if the following conditions are satisfied:
  \begin{equation}
    \forall i <j, \quad
    \left|\langle{{b_j}},{\pi_i(b_i)}\rangle \right|
    \leq \frac{1}{2}{\|\pi_i(b_i)\|^2}
    \quad \textrm{(size-reduction condition)}
  \end{equation}
  \begin{equation}
    \forall i, \quad \delta \|\pi_i({b_i})\|^2 \leq
    \left( \|\pi_{i+1}(b_{i+1})\|^2 +
        \frac{\inner{b_{i+1}}{\pi_i(b_i)}}{\|\pi_i(b_i)\|^2}
    \right)
    \quad \textrm{(Lovász condition)}
  \end{equation}
\end{definition}

In order to find a basis satisfying these conditions, it suffices to
iteratively modify the current basis at any point where one of these
conditions is violated. This yields the simplest version of the
\LLL~algorithm. As in~\cite{LLL82}, it is only defined for full-rank
sublattice of $\ZZ^d$. It was remarked by Lovász and Scarf
in~\cite{LovaszS92} that the same algorithm also works with an
arbitrary integral-valued scalar product. The method can be extended
to deal with lattices described by a generating family rather than by
a basis~\cite{Pohst87}.

\begin{algorithm}[H]
  \BlankLine
  \Parameter{$\delta \in (1/4,1)$}
  \KwIn{Initial basis $\Bas=({b_1}, \ldots, {b_d})$}
  \KwResult{A $\delta$-\LLL-reduced basis}
  \BlankLine
  $k\gets 2$\;
  Compute the $\pi_i(b_i)$'s with the
  \textsc{gso} process (Paragraph~\ref{sec:gram_schmidt})\;
  \While{$k\leq d$}{
    \lFor{$j = k-1$ \Downto $1$}{
      $b_k \gets b_k - \left\lceil
      \frac{\inner{b_k}{\pi_j(b_j)}}{\|{\pi_j(b_j)}\|^2}\right\rfloor\cdot
      b_j$
    }
    \uIf{$\delta \|\pi_{k-1}(b_{k-1})\|^2\leq
       \|\pi_{k}(b_{k})\|^2 +
        {\inner{b_{k}}{\pi_{k-1}(b_{k-1})}^2}/{\|\pi_{k-1}(b_{k-1})\|^2}$}{
      $k \gets k+1$\;
    }
    \uElse{
      Swap $b_k$ and $b_{k-1}$; Update $\pi_k(b_k)$ and $\pi_{k-1}(b_{k-1})$\;
      $k \gets \max(k-1,2)$\;
    }
  }
  \Return (${b_1}, \ldots, {b_d}$)
  \caption{The original \LLL~algorithm.}
  \label{LLL}
\end{algorithm}

\subsubsection{Decrease of the potential and complexity.}
The algorithm can only terminate when the current lattice basis is
{\LLL}-reduced. Moreover, as shown in~\cite{LLL82}, it terminates in
polynomial time when $\delta < 1$. Indeed, consider the (square of
the) product of the covolumes of the flag associated with a basis:
$\prod_{i=1}^d \|\pi(b_i)\|^{2(d-i+1)},$ which is often called its
\emph{potential}. This value decreases by a factor at least
$\delta^{-1}$ in each exchange step and is left unchanged by other
operations.
 Indeed:
\begin{itemize}
\item The flag is not modified by any operation other than swaps.
\item A swap between ${b_k}$ and ${b_{k-1}}$ only changes the sublattice
  spanned by the first $k-1$ vectors. The corresponding covolume
$\prod_{i=1}^{k-1} \|\pi(b_i)\|^{2}$ decreases by a factor at least
$\delta^{-1}$ and so does the potential.
\end{itemize}

Since the total number of iterations can be bounded by
twice the number of swaps plus the dimension of the lattice, this
suffices to conclude that it is
bounded by $\bigO{d^2\log\|B\|_\textrm{max} }$ where $B$ is the matrix of the initial
basis.

As the cost of a loop iteration is of $\bigO{d^2}$ arithmetic operations on
\emph{rational} coefficients of length at most
$\bigO{d \log\|B\|_\textrm{max}}$, the total cost in term of
arithmetic operations is loosely bounded by
$\bigO{d^6 \log^3\|B\|_\textrm{max}}$. By being more precise in the
majoration of the bit length of the integers appearing in \LLL, this
analysis can be improved. Kaltofen in~\cite{Kaltofen83} bounds the
complexity by
$$\bigO{\frac{d^5\log^2\|B\|_\textrm{max}}{d+\log\|B\|_\textrm{max}}
\mathcal{M}(d+\log\|B\|_\textrm{max})}.$$

\subsubsection{A bound on the norm of reduced elements}
\begin{proposition}
  \label{prop:norm_bound}
  Let $1/4< \delta< 1$ be an admissible \LLL~parameter.
  Let $(b_1, \ldots, b_d)$ be a $\delta$-\LLL~reduced basis of rank-$d$ lattice
  $(\Lat, \inner{\cdot}{\cdot})$. Then for any $1\leq k\leq d$:
  \[
    \Vol{b_1,\ldots,b_k}\leq \left(\delta-\frac{1}{4}\right)^{-\frac{(d-k)k}{4}}
    \Vol{\Lat}^{\frac{k}{d}}.
  \]
\end{proposition}
Note that this is an easy generalization of the bound on the norm of
$b_1$ which is given in most texts. It appears among other related
inequalities in~\cite{PatakiT08}. For completeness, a proof is given
in Appendix.

\subsubsection{Floating point representation}
The total cost of the \LLL~algorithm is dominated by the computation
to handle arithmetic on rational values.  A first idea of
De~Weger~\cite{DeWeg} to overcome this issue is to avoid the use of
denominators by multiplying all the quantities by their common
denominator. This is slightly more efficient in practice but doesn't
improve the asymptotics.  Another idea is to remark that the norms of
the rational values remain small and to try to use approximations
instead of exact values.  However, directly replacing rationals in the
\LLL~algorithm by floating-point approximations leads to severe
drawbacks. The algorithm might not even terminate, and the output
basis is not guaranteed to be \LLL-reduced.
\medskip

The first \emph{provable} floating-point version of the algorithm is due to
Schnorr in~\cite{Schnorr88},
with complexity $\bigO{d^4\log(\|B\|_\textrm{max})\mathcal{M}(d+\log\|B\|_\textrm{max})}$.
One of the key ingredients to achieve this reduction is to slightly relax the
definition of the size-reduction, in order to compensate for
the approximation errors introduced by the use of floating-point arithmetic.
We call \emph{admissible} any parameters $(\delta, \eta)$
satisfying $1/4<\delta < 1$, and $1/2 < \eta <\sqrt{\delta}$ and define:

\begin{definition}[$(\delta,\eta)$-\LLL~reduction]
  Let $(\delta, \eta)$ be admissible parameters.
  A basis $\Bas$ of a lattice is said to be
  $(\delta,\eta)$-\textbf{LLL-reduced}
  if the following condition is satisfied:
  \begin{equation}
    \forall i <j, \quad
    \left|\langle{{b_j}},{\pi_i(b_i)}\rangle \right|
    \leq \eta{\|\pi_i(b_i)\|^2}
    \quad \textrm{(Approximate size-reduction condition)}
  \end{equation}
together with the Lovász condition, which is kept unchanged from
Definition~\ref{LLLstdDef}.

\iffalse
  \begin{equation}
    \forall i, \quad \delta \|\pi_i({b_i})\|^2 \leq
    \left( \|\pi_{i+1}(b_{i+1})\|^2 +
        \frac{\inner{b_{i+1}}{\pi_i(b_i)}}{\|\pi_i(b_i)\|^2}
    \right)
    \quad \textrm{(Lovász condition)}
  \end{equation}
\fi
\end{definition}
Using naive multiplication, the cost of Schnorr's algorithm is cubic
in the size of the numbers, i.e. in $\log(\|B\|_{\textrm{max}})$. The
introduction of approximate size reduction removes the need to know
with extreme precision values close to half-integers. Instead,
approximate size reduction of such values can be achieved by rounding
either up or down in an arbitrary (possibly randomized) manner. In our
pseudo-code, we use a function called $\eta$-\textsc{Closest-Integer}
to achieve this rounding, returning an integer at distance at most
$\eta$ of the function's argument.

\subsection{The $L^2$ algorithm}

The \LL~algorithm is a variant of Schnorr-Euchner
version~\cite{SchnorrEuchner} of \LLL. By contrast with the original
algorithm, \LL~computes the \textsc{gso} coefficients on the fly as
they are needed instead of doing a full orthogonalization at the
start.  It also uses a lazy size reduction inspired by the Cholesky
factorization algorithm.  These optimizations yield an improved
lattice reduction with running time
\[ \bigO{d^5(d+\log(\|B\|_\textrm{max}))\log(\|B\|_\textrm{max})}. \]

As usual in lattice reduction, while performing the Gram-Schmidt
orthogonalization of $\Bas$, we also compute
\emph{QR}-decomposition of $B$ into $B^*\cdot M$ where $B^*$ is the matrix
representing the $\left(\pi_i(b_i)\right)_{1\leq i \leq d}$,  and $M$
is the upper unitriangular
matrix, whose coefficients with $j\geq i$ are $M_{i,j} =
\frac{\inner{b_j}{\pi_i(b_i)}}{\|\pi_i(b_i)\|^2}$. Thus, the Gram
matrix associated to the basis, i.e., $G = B^TB$ satisfies:
\[
  G = M^T\cdot {B^*}^T \cdot B^* \cdot M  = M^T\cdot D\cdot M
\]
where  $D$ is a diagonal matrix whose entries are
${\|\pi_i(b_i)\|^2}$.
We denote by $R$ the matrix $D\cdot M$, and thus have
$ G = R^T\cdot M = M^T\cdot R.$

We give the pseudo-code of the Lazy Size-Reduction procedure as Algorithm~\ref{SredLazy}
and of the \LL~algorithm as Algorithm~\ref{LL}. Both use classical
formulas relating $R$, $M$ and $B^*$ to perform the computations.

\begin{algorithm}[p]
  \BlankLine
  \KwIn{Initial basis $\Bas = ({b_1}, \ldots, {b_d})$, with $G$, $R$
    and $M$. An integer
  $1\leq k \leq d$.}
  \KwResult{Size-reduces $b_k$, updates $G$,
    $R$, $M$ and returns $s^{(k)}$ }
  done $\gets \False$\;
  \While{{\rm done} $= \False$}{
    \For{$j=1$ \KwTo $k-1$}{
      ${R_{k,j}} \gets G_{k,j}$;
      \lFor{$i =1$ \KwTo $j-1$}{
        $R_{k,j} \gets R_{k,j} - M_{j,i}R_{k,i}$}
      ${M_{k,j}} \gets {R_{k,j}}/{R_{j,j}}$;
    }
    ${s}_1^{(k)} \gets G_{k,k}$;
    \lFor{$j=2$ \KwTo $k$}{
      ${s}_j^{(k)} \gets {s_{j-1}^{(k)}}
      - {M}_{k,j-1}\cdot {R_{k,j-1}}$}
    ${R_{k,k}} \gets {s_k^{(k)}}$\;
    \lIf{$(\max_{j<k} |M_{k,j}|) \leq \eta$}{done $\gets \True$}
    \Else{
      \For{$i=k-1$ \Downto $1$}{
        $X_i \gets$ $\eta$-\textsc{Closest-Integer}$(M_{k,i})$\;
        \lFor{$j=1$ \To $i-1$}{$M_{k,j} \gets M_{k,j} - X_iM_{i,j}$}
      }
      $b_k \gets b_k - \sum_{i=1}^k X_i b_i$;
      Update $G$ accordingly\;
    }
  }
  \caption{The lazy size reduction algorithm, $\eta$-\textsc{LazyRed}.}
  \label{SredLazy}
\end{algorithm}

\begin{algorithm}
  \BlankLine
  \Parameter{ $\delta \in (1/4,1), \eta \in (1/2, \sqrt{\delta}$).}
  \KwIn{Initial basis $\Bas=({b_1}, \ldots, {b_d})$}
  \KwResult{A $(\delta,\eta)$-\LLL-reduced basis}
  \BlankLine
  Compute $G = G({b_1}, \cdots, {b_d})$ in exact integer arithmetic\;
  ${R}_{1,1} \gets G_{1,1}$\;
    $k \gets 2$\;
    \While{$k \leq d$ }{
      Apply size reduction $\eta$-\textsc{LazyRed}$(k)$\;
      $k' \gets k$\;
      \lWhile{($k \geq 2$ \And $\delta {R}_{k-1,k-1} > {s}_{k-1}^{k'})$}{
        $k\gets k-1$
        }
        $R_{k,k} \gets s_{k}^{k'}$\;
        \If{$k\neq k'$}{
          \lFor{$i=1$ \KwTo $k-1$}{
            $M_{k,i} \gets M_{k',i}$;
            $R_{k,i} \gets R_{k',i}$
        }
        $R_{k,k} \gets s_{k}^{k'}$\;
          Insert ${b_{k'}}$ at pos $k$ (before $b_k$) and update matrix $G$ accordingly\;
        }

        $k\gets k+1$\;
      }
  \Return (${b_1}, \cdots, {b_d}$)
  \caption{The $L^2$ Algorithm.}
  \label{LL}
\end{algorithm}

\subsubsection{Precision required.}
\label{LLprecision}
The precision required by the \LL-Algorithm is \[
{d\log\left(\frac{(1+\eta)^2}{(\delta-\eta)^2}+\epsilon\right)+\littleO{d}}\]  bits for any
$\epsilon > 0$, i.e., almost linear in the dimension of the lattice.  Moreover,
as discussed in~\cite{NguyenS09}, it appears that---even though this
bound can be shown to be sharp by specific examples---experiments indicate
that the number of bits required \emph{on average}
is, in fact, lower.

This phenomenom is well-known and is often used in existing algorithms
and softwares in the form of a compute-and-verify paradigm. For
example, this is default strategy of the well-known
FPLLL~\cite{fplll}. It relies on the fact that verifying that a
lattice basis is indeed reduced is much less costly than the reduction
itself, as shown in~\cite{Villard07}. In addition, it is necessary to
take several conservative measures in order to prevent the
implementation to enter potentially infinite loops.

The approach we propose deviates from this paradigm. Instead of
guaranteeing the end-result, we want to make sure that the whole
computation follows the mathematical definition of the
algorithm. With low-precision approximations, it is unclear how this
could be done. However, interval-arithmetic offers a neat solution to
achieve this goal.
 
\section{Interval Arithmetic and its certification property}
Interval arithmetic is a representation of reals by intervals that
contain them.  For instance, one can specify a value $x$
with an error $\epsilon$ by giving an interval of length $\epsilon$
containing $x$.  For example, the constant
$\pi$ can be represented with an error of~$10^{-2}$ by the interval
$[3.14, 3.15]$. Interval arithmetic is crucial in the context of
\emph{certified} numerical computations, where reals can only be
represented with finite precision. For more details, the interested
reader can consult an extensive reference, such
as~\cite{InterBook}.\medskip

In the following, we denote
by $\inter{x}$ a closed interval $[\inter{x}^-, \inter{x}^+]$.
We define its \emph{diameter} as the 
positive real $\inter{x}^+ - \inter{x}^-$ and its \emph{center} as the real 
$\frac{1}{2}(\inter{x}^+ + \inter{x}^-)$. 

Given a real-valued function $f(x_1,\ldots,x_n)$ an
interval-arithmetic realization of $f$ is an interval-valued function
$F$ such that the interval $F(\inter{x_1}, \ldots, \inter{x_n})$
contains all the values $f(x_1,\ldots,x_n)$ for $(x_1,\ldots,x_n)$ in
$\inter{x_1}\times \cdots \times \inter{x_n}$.

If $F$ always returns the smallest possible interval, it is called a
\emph{tight} realization, otherwise it is called \emph{loose}. In
practice, tight realizations can only be achieved in very simple
specific cases. However, even a loose realization can suffice to
certify the correctness of a computation.

Another important property of interval arithmetic is that it can be
used to compare numbers in a certified way, as long as the intervals
that represent them are disjoint.

\subsection{Some useful interval-arithmetic realizations}
\subsubsection{Integral representation of fixed length}
\label{sec:intgral_representation}
A first convenient way to represent reals at finite precision is to
use integers as an approximate representation.

\begin{definition}[Integral representation of reals]
  \label{def:integral_representation}
  Let $x\in\RR$ be an arbitrary real number and $n\geq 0$ a non-negative integer.  
  Define \emph{an integral representation at accuracy\footnote{
      We use here the denomination of ``accuracy'' instead of
      ``precision'' to avoid 
      confusions with the floating-point precision as defined in 
    paragraph~\ref{sec:fp_representation}.}~
  $n$} as an interval of diameter $2$:
  \[
    \inter{x}_n = \left[X_n-1, X_n+1 \right]
  \]
together with a guarantee that $2^nx$ belongs to $\inter{x}_n$.
\end{definition}

This representation is very compact, since it only requires to store
the center $X_n$ of the interval using $n+\lceil\log x\rceil$
bits. However, computing with this form of representation is not
convenient. As a consequence, we only use it to represent immutable
values and we convert to a different representation for
computations. The reason for using the interval
$\left[X_n-1, X_n+1 \right]$ of diameter $2$ rather than
$\left[X_n-1/2, X_n+1/2 \right]$ (of diameter $1$) is that when
$2^{n}x$ is very close to a half-integer, it remains possible to
easily provide a valid value for $X_n$ without computing extraneous
bits of the representation of $x$.

\subsubsection{Fixed-point representations}
In the context of lattice reduction, it is useful to compute linear
combinations with exact integral coefficients. In order to do that
with approximate values initially given by centered integral
representation, it is possible to use a fixed-point representation. 

\begin{definition}[Fixed point representation of reals]
  \label{def:fixed_point_representation}
  Let $x\in\RR$ be an arbitrary real number and $n\geq 0$ a
  non-negative integer.  Define a \emph{fixed-point representation at
    accuracy $n$} of radius $\delta$ as an interval:
  \[
    \inter{x}_n = \left[X_n-\delta, X_n+\delta \right]
  \]
together with a guarantee that $2^nx$ belongs to $\inter{x}_n$.
\end{definition}

It is easy to add or subtract such intervals by doing the computation
on the center and by adding the two radii. It is also easy to multiply
by an exact integer by multiplying the center by the integer and the
radius by its absolute value. Integral representations are a special
case of fixed-point representations, with radius equal to $1$.

\subsubsection{Floating-point representation.}
\label{sec:fp_representation}
Another way to handle real values is to use floating point
representations of the two bounds of each interval. For example, if we
denote by $ \lfloor x\rfloor_n$ and $\lceil x \rceil_n$ respectively
the largest floating-point number below $x$ and the lowest
floating-point number above $x$ written with $n$ bits, the tightest
floating-point representation of $x$ with $n$ bits of precision is the
interval
$I_n(x) = \left[ \lfloor x\rfloor_n, \lceil x \rceil_n\right]$.

With such a representation, it becomes possible to create a
realization of the elementary operations by using careful rounding
when computing approximations of the bounds of the resulting interval,
as shown in Figure~\ref{inter_arith}. When speaking of the precision
of such a representation, we simply refer to the common floating-point
precision of the upper and lower bounds.

Once the elementary operations are available, they can be used to
implement certified versions of any function that can classically be
computed with floating point arithmetic. 

\begin{figure}
\small \[
  \begin{aligned}
    \left[\inter{x}^-, \inter{x}^+\right] + \left[\inter{y}^-, 
    \inter{y}^+\right] &= \left[\inter{x}^-+^-\inter{y}^-, \inter{x}^++^+ 
    \inter{y}^+\right]         \\
    \left[\inter{x}^-, \inter{x}^+\right] - \left[\inter{y}^-, 
    \inter{y}^+\right] &= \left[\inter{x}^--^-\inter{y}^-, \inter{x}^+-^+ 
    \inter{y}^+\right]         \\
    \left[\inter{x}^-, \inter{x}^+\right] \times \left[\inter{y}^-, 
    \inter{y}^+\right] &= \left[\textrm{min}^-(\rho), \textrm{max}^+ (\rho) 
  \right] \qquad \text{where} \quad \rho =  \inter{x}^-\inter{y}^-,
    \inter{x}^+\inter{y}^-, \inter{x}^-\inter{y}^+, \inter{x}^+\inter{y}^+        
    \\
    \left[\inter{x}^-, \inter{x}^+\right]^{-1} &=
    \left[\textrm{min}^-\left(\frac{1}{\inter{x}^+}, \frac{1}{\inter{x}^-}\right),
    \textrm{max}^+\left(\frac{1}{\inter{x}^+},
    \frac{1}{\inter{x}^-}\right)\right]
  \end{aligned}
\]
\hrulefill\\
\emph{\small{$+^+$, $+^-$ are here respectively the $+$ operator with
   rounding up or down.\\
The same goes for the $-^+,-^-, \min^-, \max^+$ 
operators.}}
\caption{Basic arithmetic operators in Interval Arithmetic}
\label{inter_arith}
\end{figure}

\section{Approximate lattices}
The need to reduce lattices given by approximations, especially for
number-theoretic applications as been known for long. In particular,
Buchmann gives in~\cite{Buchmann94} a bound on the required precision
to achieve this goal by using a direct approximation of the input
basis. However, this bound is computed in terms of a quantity called
the defect that can be very large and also involves the first minimum
of the lattice.

Using interval arithmetic, it becomes possible to get finer control on
the precision required to perform the lattice reduction, even with
approximate lattices.

\subsection{Approximate representation of a positive-definite matrix}
\label{sec:matrix_representation}
A matrix with real entries  can easily be represented with the integral
representation from Definition~\ref{def:integral_representation},
using the same accuracy for all of its entries.

\begin{definition}[Matrix integral representation]
\label{def:matintrep}
  Let $A=(a_{i,j})_{i,j}\in\RR^{d\times d}$ 
  be an arbitrary real matrix of dimension $d$ and $n>0$ be a fixed positive integer.
  A matrix of intervals 
  \[ \inter{A}_n = 
      (\inter{a_{i,j}}_n)_{(i,j)\in [1\,\cdots\,d]^2},
  \]
where each $\inter{a_{i,j}}_n$ is an integral representation of
$a_{i,j}$  is said to \emph{integrally represent~$A$ at accuracy
  $n$}.
\end{definition}

We may omit the subscript $n$ when the accuracy is clear from the 
context. Given a matrix $A$,  and a matrix $B\in\inter{A}_n$, there
exists a unique $d\times d$ matrix $\Delta$ with entries in
$[-2, 2]$ such that $B= 2^n A+\Delta$.

In particular, we may apply this representation to symmetric
matrices. In that case, we obtain the following useful lemma:

\begin{lemma}
  \label{lem:minimal_eigenvalue}
    Let $S=(s_{i,j})_{i,j}\in\mathcal{S}_d(\RR)$ be a symmetric
    matrix of dimension $d$ and $\inter{S}_n$ an integral
    representation of $S$ at accuracy $n$. Then, for any symmetric
    matrix~$S'$ in $\inter{S}_n$, we have:
    $$2^n\lambda_d(S)-2d\leq \lambda_d(S') \leq 2^n\lambda_d(S)+2d,$$
    where $\lambda_d(T)$ denotes the smallest eigenvalue of a $d$-dimensional
    symmetric matrix~$T$.
\end{lemma}

\begin{proof}
  This is a direct consequence of Weyl's inequalities for Hermitian
  matrices and of the relation $S'=2^nS+\Delta$, where $\Delta$ is
  real symmetric with entries in $[-2,2]$. Note that the eigenvalues
  of $\Delta$ all belong to $[-2d,2d]$.
\end{proof}

\subsection{Representation of lattices}
\label{sec:approx_representation}

In order to represent arbitrary lattices, we first need a description
of their ambient space. We simply describe the ambient space $V$ of
dimension $d$  by providing a basis $\gamma = 
(\gamma_1, \ldots, \gamma_d)$. Then, the scalar product
$\inner{\cdot}{\cdot}$ on $V$ can be encoded by a Gram matrix 
$\Gram_\gamma =  \left(\inner{\gamma_i}{\gamma_j}\right)_{(i,j)\in [1\,\cdots\,d]^2}$.

When the Gram matrix $\Gram_\gamma$ is integral, this already is a
standard description of the lattice $\Gamma$ spanned by $\gamma$. This
representation appears in particular
in~\cite[Proposition~2.5.3]{Cohen93}. We now extend this in order to
represent bases and generating families of arbitrary sublattices of
$\Gamma$.  Let $\Lat$ be a rank $r\leq d$ sublattice of $\Gamma$ given
by a generating family $\ell = \left(\ell_1, \ldots, \ell_p
\right)$. Since any vector in $\ell$ belongs to $\Gamma$, it can be
expressed with integral coordinates in the basis $\gamma$. As a
consequence, we can represent $\ell$ by a $p \times d$ integral matrix
$L$. Moreover, the knowledge of $\Gram_\gamma$ allows us to easily
compute the scalar product of any pair of vectors in $\Lat$.

All this leads to the following definition:
\begin{definition}[Approximate representation of a lattice]
  \label{def:approx_representation}
  Let $\Gram_\gamma$ and $L$ be as above and $n$ be a non-negative integer. Denote
  by $G$ the matrix of centers of an integral representation 
  $\inter{\Gram_\gamma}_n $at accuracy $n$ of the Gram matrix $\Gram_\gamma$.
  Then the pair $(G, L)\in\ZZ^{d\times d}\times\ZZ^{p\times d}$ of integral 
  matrices is said to \emph{represent at accuracy $n$} the lattice $\Lambda$ in 
  the basis $\gamma$ of $\Gamma$.
\end{definition} 

\subsubsection{Computation of the inner product in Interval Arithmetic.}
\label{sec:inner_product_computation}
Let $a$ and $b$ be two vectors of $\Lat$ described by their vectors $A$
and $B$ of coordinates in the basis $\gamma$. We know that:
$$
\inner{a}{b} =  A^T\cdot \Gram_\gamma \cdot B.
$$

Thus:
$$
2^n \inner{a}{b} = A^T\cdot G \cdot B + A^T\cdot \Delta \cdot B,
~\mbox{where}~
|A^T\cdot \Delta \cdot B| \leq \left(\sum_i |A_i|\right)\left(\sum_i |B_i|\right).
$$
This directly gives an interval representation of $\inner{a}{b}$.

\subsection{Lattice reduction of approximate lattices}
\label{sec:computational_setting}

Suppose now that the Gram matrix
$\Gram_\gamma = \left(\inner{\gamma_i}{\gamma_j}\right)_{(i,j)\in [1\,\cdots\,d]^2}$ 
representing the inner product of the ambient space
$\Gamma \otimes_\ZZ \RR$ in the basis $\gamma$ is given indirectly by
an algorithm or an oracle $\mathcal{O}_\gamma$ that can compute each
entry at any desired accuracy. We can restate the definition of a 
reduced basis in this framework as:

\begin{definition}[$(\delta, \eta)$-\LLL~reduction]
  \label{def:LLL_framework}
  Let $(\delta, \eta)$ be admissible \LLL~parameters. Given an
  integral matrix $L \in \ZZ^{p\times d}$ which describes the vectors
  of a basis of a lattice $\Lat$ in the basis $\gamma$, we
  say that $(\Gram_\gamma,L)$ is a $(\delta,\eta)$-\LLL~reduced basis 
  of $\Lat$ if and only if there
  exists an $n_0>0$ such that for any $n\geq n_0$ there exists a pair
  $({\inter G}_n,L)$, where $ {\inter G}_n$ is an integral representation of
  $\Gram_\gamma$ at accuracy $n$, which is a $(\delta,\eta)$-\LLL~reduced
  basis.
\end{definition}

The computational problem associated with reduction theory can then be 
written as:

\begin{problem}[Lattice Reduction for approximate representation]
  Let $\delta, \eta$ be admissible \LLL~parameters. Given as input an
  algorithm or oracle to compute $\Gram_\gamma$ at arbitrary precision
  and an integral matrix $L \in \ZZ^{p\times d}$ that describes the
  vectors of a generating family of a lattice $\Lat$ in the basis
  $\gamma$: find a basis $L'$ of $\Lat$ such that $(\Gram_\gamma,L')$
  is a $(\delta,\eta)$-\LLL~reduced basis in the sense of
  Definition~\ref{def:LLL_framework}.
\end{problem}

Note that using interval arithmetic it suffices to check the
$(\delta,\eta)$-\LLL~reduction condition at accuracy $n_0$ to be sure
it holds at any larger accuracy. Indeed, an integral representation
that satisfies the condition can be refined into a more precise
integral representation by scaling up the integer representing the
center by an adequate power of two. This refined representation
continues to satisfy the condition.

\subsubsection{Accuracy of representation and space complexity.}
\label{sec:optimization_space}
Let $({\inter G}_n,L)$ be an integral representation of $\Lat$, at accuracy $n$.
Then, the magnitude of the entries of $G$ is $2^n$ times the magnitude of 
the entries of $\Gram_\gamma$. Thus, ${\inter G}_n$ can be encoded
using 
$\bigO{d^2(n+\log \|\Gram_\gamma\|_\textrm{max})}$ bits.

\section{Generalized LLL~reduction with Interval Arithmetic}
In this Section, we adapt lattice reduction algorithms to our
setting. More precisely, we represent the information related to
Gram-Schmidt vectors by interval arithmetic using a floating-point
representation as described in
Section~\ref{sec:fp_representation}. For the representation of the
lattice itself, we consider two cases: either the underlying Gram
matrix is integral, or it is given by an approximate integral
representation as in Section~\ref{sec:matrix_representation}.  In the
latter case, our algorithm also asks for representations with higher
accuracy until it is sufficient to yield a reduced basis for the given
lattice. The canonical case with the standard Euclidean scalar product
is achieved by setting the Gram matrix to the (exact) identity matrix.

\subsection{Interval Arithmetic $L^2$ reduction with fixed precision.}
\label{sec:algorithm}

We first consider the simplified case where the lattice representation
is fixed. It can be either exact or approximate with a given
accuracy. In both cases, we fix a basis
$\gamma = (\gamma_1, \ldots, \gamma_d)$ and a representation of a
lattice $\Lat$ in this basis. It is respectively an exact integral
representation $(G,L)$ or an approximate representation
$({\inter G}_n,L)$ at accuracy $n$ of $(\Gram_{\gamma},L)$.

\subsubsection{Using Interval Arithmetic in \LLL.}
We now modify the \LL~algorithm of~\cite{NguyenS09} in a few relevant
places to make use of interval arithmetic instead of floating-point
arithmetic for the Gram-Schmidt-related values.  Since the description
of the lattice $\Lat$ is already using intervals, it seems natural to
use interval arithmetic in the lattice reduction algorithm. For
completeness, when the input Gram matrix is exact, we make the updates
to the Gram-Schmidt orthogonalized matrix used by \LLL{} explicit in
the algorithm (except the simple displacements).  This also emphasizes
a subtle difference with the case of an approximate input Gram
matrix. Indeed, in that case, we update the \GSO-values but recompute
the errors rather than relying on the interval arithmetic to do
it. This is important to gain a fine control on the error growth
during updates.

In addition, when using the technique from~\cite{Pohst87} to be able
to deal with lattices given by a generating family instead of a basis,
we make a slightly different choice than in~\cite{NguyenS09}. Instead
of moving the zero vectors that are encountered during the computation
during the reduction to the start of the basis, we simply remove them.
Note that with an approximate matrix, if we discover a non-zero vector
whose length is given by an interval containing $0$, it is not
possible to continue the computation. This means that the accuracy of
the input is insufficient and we abort. The core modification with
interval arithmetic appears while testing the Lovász condition. If it
is not possible to decide whether the test is true or false because
of interval overlap, we also abort due to lack of precision. To be
more precise, when testing the Lovász condition, we also need to check
that the corresponding $\mu$ coefficient is indeed smaller than
$\eta$. The reason for this is that, when called with insufficient
precision, the Lazy reduction routine may fail to ensure that
property.

In addition, if a negative number occurs when computing the norm of a
vector, it means that the given Gram matrix is not positive-definite
and the algorithm returns an error accordingly.

\begin{algorithm}
  \BlankLine
  \KwIn{
    Initial basis $L = ({L_1}, \ldots, {L_d})$,
    precomputed (internal) Gram matrix $Gram$, interval matrices
  $\inter{R}$ and $\inter{M}$, an integer $1\leq k \leq d$.}
  \KwResult{Size-reduce the $k$-th vector of $L$ and update the
  Gram matrix $Gram$.}
  \BlankLine
  done $\gets \False$\;
  \While{{\rm done} $= \False$}{
    \For{$j=1$ \KwTo $k-1$}{
      $\inter{R_{k,j}} \gets \textsc{ConvertToFPinterval}({Gram_{k,j}})$\;
      \lFor{$i =1$ \KwTo $j-1$}{
        $\inter{R_{k,j}} \gets \inter{R_{k,j}} -
        \inter{M_{j,i}}\inter{R_{k,i}}$ }
      $\inter{M_{k,j}} \gets \inter{R_{k,j}}/\inter{R_{j,j}}$;
    }
    $\inter{{s}_1^{(k)}} \gets  \textsc{ConvertToFPinterval}({Gram_{k,k}})$\;
    \lFor{$j=2$ \KwTo $k$}{
      $\inter{{s}_j^{(k)}} \gets \inter{s_{j-1}^{(k)}}
      - \inter{{M}_{k,j-1}}\cdot \inter{R_{k,j-1}}$ }
    $\inter{R_{k,k}} \gets \inter{s_k^{(k)}}$\;
    $\inter{\tau} \gets (\max_{j<k} \inter{M_{k,j}})$\;
    ret $\gets (\inter{\tau} \leq \eta)$\;
    \lIf{ret $\neq \False$}{ done $\gets \True$}
    \Else{
      \For{$i=k-1$ \Downto $1$}{
        $X_i \gets$ $\eta$-\textsc{IntervalClosestInteger}$(\inter{M_{k,i}})$\;
        \lFor{$j=1$ \To $i-1$}{
          $\inter{M_{k,j}} \gets \inter{M_{k,j}} -
          X_i\inter{M_{i,j}}$}
        $L_k \gets L_k - X_i L_i$\;
        \tcp{Update the Gram matrix accordingly}
        ${Gram_{k,k}} \gets {Gram_{k,k}} - 2X_i {Gram_{k,i}}+X_i^2 {Gram_{i,i}}$\;
        \lFor{$j=1$ \To $i$}{${Gram_{k,j}} \gets {Gram_{k,j}} - X_i {Gram_{i,j}}$}
        \lFor{$j=i+1$ \To $k-1$}{${Gram_{k,j}} \gets {Gram_{k,j}} - X_i {Gram_{j,i}}$}
        \lFor{$j=k+1$ \To $d$}{${Gram_{j,k}} \gets
          {Gram_{j,k}} - X_i {Gram_{j,i}}$}
      }
    }
  }
  \caption{The (interval) lazy size reduction algorithm, $\eta$-\textsc{ILazyRed}.}
  \label{inter_SredLazy}
\end{algorithm}
\begin{algorithm}
  \BlankLine
  \Parameter{ $\delta \in (1/4,1), \eta \in (1/2, \sqrt{\delta})$
    admissible
  \LLL~parameters, $\ell\in\NN$ the internal precision used for
  floating-point
  representation.}
  \KwIn{Exact representation $(G,L)$ or approximate representation
    $(\inter{G_n},L)$ of a lattice given by $p$  generating vectors in dimension $d$.}
  \KwResult{A $(\delta,\eta)$ \LLL-reduced basis $L'$ (with $\dim(L)$ vectors).}
  \BlankLine
  $k \gets 2$ \;
  \tcp{Compute the Gram matrix of the basis represented by $L$}
    \For{$i=1$ \To $p$ {\bf for} $j=1$ \To $i$}{
      \lIf{Exact}{$GramL_{i,j} \gets L_i^T G L_j$}
      \lElse{
      $GramL_{i,j} \gets $ Interval of center  $L_i^T G_n L_j$
      and radius $\|L_i\|_1\|L_j\|_1$
}}

  $\inter{R_{1,1}} \gets \textsc{ConvertToFPinterval}(GramL_{1,1})$\;
  \While{$k \leq p$ }{
    \tcp{Size-reduce $L_k$ with interval on the family
    $(L_{1},\ldots, L_{k-1})$}
  $\eta$-\textsc{ILazyRed}(k, {\it Exact})\;
      {{\bf if} {\it Exact = \False} {\bf then} \For{$j=1$ \To $k$}{
    Update radius of $GramL_{k,j}$ to $\|L_k\|_1\|L_j\|_1$ (rounded up
  with $\ell$ significant bits)}}
    $k' \gets k$\;
    \While{$k \geq 2$}{
     ret $\gets \left(\inter{M_{k',k-1}}\leq \eta\right)\mbox{~and~}\left(\inter{\delta} \cdot  \inter{R_{k-1,k-1}} >
      \inter{s_{k-1}^{(k')}}\right)$\;
      \lIf{ret = \True}{$k \gets k-1$}
      \lElseIf{ret = \False}{\Break}
      \lElse{
       \Return ErrorPrecision
      }
    }

    \uIf{$k\neq k'$}{
      \lFor{$i=1$ \KwTo $k-1$}{
        $\inter{M_{k,i}} \gets \inter{M_{k',i}}$;
        $\inter{R_{k,i}} \gets \inter{R_{k',i}}$
      }
      $\inter{R_{k,k}} \gets \inter{s_{k}^{k'}}$;
      $L_{tmp} \gets L_{k'}$;
      \lFor{$i=k'$ \Downto $k+1$}{
        $L_{i} \gets L_{i-1}$
      }
      $L_{k} \gets L_{tmp}$;  Move values in ${GramL}$ accordingly\;
    }
    \Else{
      $\inter{R_{k,k}} \gets \inter{s_{k}^{(k')}}$\;
      \lIf{$0 \in \inter{R_{k,k}}$ and $L_k \neq 0$}{
          \Return ErrorAccuracy
      }
      \lIf{$\inter{R_{k,k}}<0$}{
          \Return ErrorNonPosDefinite
      }

    }

    \If{$L_k = 0$}{
      \lFor{$i=k$ \To $p-1$}{ $L_i \gets L_{i+1}$ }
      $p \gets p-1$; $k \gets k-1$; Move values in ${GramL}$ accordingly\;
    }
    $k\gets \max(k+1, 2)$\;
  }
  \Return ($L$)
  \caption{The \TLL~Algorithm.}
  \label{alg:LL_tilde}
\end{algorithm}

\subsubsection{Internal precision in the exact-input case.}
\label{sec:master_precision}

For the classical \LL~algorithm, Section~\ref{LLprecision} states that
the precision that is needed for the computations only depends on the
dimension of the lattice. It is natural to ask a similar question
about the algorithm \TLL: can the required internal accuracy be
bounded independently of the entries appearing in the matrices
$G$ and $L$. When $G$ is exact, i.e., integral, the adaptation is
straightforward and we obtain the following result.

\begin{theorem}
  \label{thm:master_precision}
  Let $(\delta,\eta)$ be admissible \LLL{} parameters. Let
  $c>\log\frac{(1+\eta)^2}{\delta-\eta^2}$ and let
  $(\Lambda, \inner{\cdot}{\cdot})$ denote a rank-$d$ lattice, exactly
  described by the pair $(G,L)$. Let $B$ denotes the maximum
  entry in absolute value in $L^TGL$. Then, the
  \TLL~of Figure~\ref{alg:LL_tilde} used with $\ell=cd+\littleO{d}$ outputs
  a $(\delta,\eta)$-{\LLL}-reduced basis
 in time $\bigO{d^3\log{B}(d+\log{B})\Mult{d}}.$ Furthermore, if $\tau$
 denotes the number of main loop iterations, the running time is
$\bigO{d(\tau+d\log{dB})(d+\log{B})\Mult{d}}.$
\end{theorem}

In fact, the bound on $\ell$ is made explicit
in~\cite{NguyenS09}. More precisely, it states that for any arbitrary
$C>0$ and an $\epsilon\in]0,1/2]$, it suffices to have:
$$
\ell \geq
10+2\log_2{d}-\log_2{\min(\epsilon,\eta-1/2)}+d(C+\log_2\rho) \quad\mbox{where~}\rho=\frac{(1+\eta)^2+\epsilon}{\delta-\eta^2}.
$$
For example, choosing $C=\epsilon=\eta-1/2$ it suffices to have:
$$
\ell \geq T(d,\delta,\eta)=
10+2\log_2{d}-\log_2{(\eta-1/2)}+(\eta-1/2+\log_2\rho)\,d.
$$
When $\delta$ is close to $1$ and $\eta$ to $1/2$, the constant before
$d$ becomes smaller than $1.6$.

\subsubsection{Dealing with approximate inputs.}
When dealing with lattices given in an approximate form, i.e., by a
representation $({\inter G}_n,L)$ at accuracy $n$ of
$(\Gram_{\gamma},L)$, the analysis of the algorithms differs in three
main places:
\begin{itemize}
\item When bounding the number of rounds $\tau$, we can no longer
  assume that the potential is an integer. As a consequence, in order
  to keep a polynomial bound on $\tau$, we need to provide a lower
  bound on the possible values of the potential, rather than rely on
  the trivial lower bound of $1$ for an integral-valued potential.
\item Since the notion of {\LLL}-reduction is only well-defined for a
  positive definite $G$, we need to make sure that ${\inter G}_n$ is
  positive-definite during the algorithm. Otherwise, it should output
  an error; Algorithm~\ref{alg:LL_tilde} returns an error that
  ${\inter G}_n$ is incorrect whenever it encounters a vector with a
  negative norm.
\item When ${\inter G}_n$ is approximate, the scalar products between
  lattice vectors can no longer be exactly computed. Thus, we need to
  able to make sure that the errors are small enough to be compatible
  with the inner precision used for Gram-Schmidt values. At first
  glance, this might seem easy. However, when using update formulas to
  avoid recomputation of scalar products, the estimates on errors
  provided by interval arithmetic can grow quite quickly. In fact,it
  would prevent the update strategy from working. The key insight is
  to remark that since the centers of the intervals are represented by
  integers, any computation on them is exact and we can use update
  formulas to compute them. However, it is essential to recompute the
  radii of the intervals, i.e., the errors, to prevent them from
  growing too quickly.
\end{itemize}

\subsubsection*{Number of rounds}
Since interval arithmetic allows up to emulate exact computations as
long as no failures are detected, we can analyze the number of rounds
by assuming that all computations on non-integral values are done
using an exact arithmetic oracle. In this context, the number of
rounds can be studied by considering the potential as usual. Remember that
the initial setting where \LLL~operates on a basis the potential is
defined as
$$
D(B)=\prod_{i=1}^{d}\Vol{B_{[1\ldots i]}}.
$$
The key argument is that it decreases by a multiplicative
factor whenever an exchange is performed.

However, in our context, the starting upper bound and the ending lower
bound are different from the integer lattice setting. The initial
upper bound needs to account from the presence of the positive
definite matrix. So if the lattice is described by a pair $(\Gram_{\gamma},L)$ the
upper bound becomes:
$$
D(B)^2 \leq \left(d^2 \|\Gram_{\gamma}\|_\textrm{max} \|L\|_\textrm{max}^2\right)^{d(d+1)/2}.
$$
More importantly, it is no longer possible to claim that the potential
is an integer. Instead, we derive a lower bound by considering the
smallest eigenvalue of $\Gram_{\gamma}$ and find:
$$
D(B)^2\geq \lambda_d(\Gram_{\gamma})^{d(d+1)/2}.
$$

As a consequence, if we let $\tau$ denote the number of rounds of
the algorithm, we can conclude that:
$$
\tau \leq \bigO{d^2 \left( \log(\|L\|_\textrm{max} \right)
+ \log\left(\|\Gram_{\gamma}\|_\textrm{max}/\lambda_d(\Gram_{\gamma}) \right)+\log(d)}.
$$

When the lattice is given by a generating family $L$ rather than a
basis $B$,
we need a slightly different invariant. Following~\cite{NguyenS09}, we
define $d_i$ to be the product of the first $i$ non-zero values
$\|b^*_j\|$. Note that they are not necessarily consecutive, since
zeroes may occur anywhere. We then let:
$$
D'(L)=\left(\prod_{i=1}^{\dim{L}}d_i\right) \cdot \left(\prod_{i, b^*_i=0}2^i\right).
$$
This generalized potential is needed for the proof of
Theorem~\ref{thm:master_approx}. Note that, for lattices given by a
basis, the two definitions coincide.

\subsubsection*{Necessary accuracy for the scalar products}
In order to preserve the correctness of the algorithm when computing
with internal precision $\ell$, we need to check that all conversions
of scalar product values, using the calls to
\textsc{ConvertToFPinterval} in Algorithms~\ref{inter_SredLazy} and
\ref{alg:LL_tilde}, have sufficient precision. For a pair of lattice
elements, described by vectors $L_i$ and $L_j$, the relative precision
on the value of their scalar product is:
$$
\frac{\|L_i\|_1\|L_j\|_1}{|L_i^T G_n L_j|}.
$$
When the vectors are close to orthogonal with respect to the scalar
product given by $G_n$, the error can be arbitrarily large.
However, by carefully following the analysis of Theorem~3
in~\cite[Section 4.1]{NguyenS09}, we can show that this Theorem
remains true in our context. This suffices to ensure the correctness
part of Theorem~5 of~\cite{NguyenS09}. The first check is to verify
that quantity called $err_1$ in the proof of the Theorem remains upper
bounded by $2^{-\ell}$. Since the value is defined as the error on the
scalar product of the vectors number $i$ and $1$ divided by the norm
of the first vector, we have:
$$
err_1 \leq \frac{\|L_i\|_1\|L_1\|_1}{|L_1^T G_n L_1|}\leq
\frac{\max_{i}{\|L_i\|_1^2}}{\lambda_d(G_n)}
\leq \frac{d\max_{i}{\|L_i\|^2}}{\lambda_d(G_n)}
\leq \frac{d\max_{i}{\|b_i\|^2}}{\lambda_d(G_n)^2}.
$$

Thus:
$$
err_1 \leq \frac{d^3 \|G_n\|_\textrm{max}\|L\|^2_\textrm{max}}{\lambda_d(G_n)^2} \leq
\frac{d^3 \left(2^n\|\Gram_{\gamma}\|_\textrm{max}+1\right)\|L\|^2_\textrm{max}}
{\left(2^n\lambda_d(\Gram_{\gamma})-2d\right)^2}.
$$

As a consequence, it suffices to have:
$$
n \geq \ell+\bigO{\log(\|L\|_\textrm{max})
+ \log\left(\|\Gram_{\gamma}\|_\textrm{max}/\lambda_d(\Gram_{\gamma})\right)
+\log(d)}.
$$

\subsubsection*{\LL with approximate inputs.}
To complete the above properties on the number of rounds and necessary
accuracy, it suffices to remark that the only additional line of code
in the approximate \LL is the recomputation of interval radii on
line~10. Since it suffices to know the $\ell$ high-order bits of the
values, this recomputation can fully be done using arithmetic on
$\ell$. Indeed, during the computations of $\|L_i\|_1$ no cancellation
occurs. As a consequence, we get the following adaptation of
Theorem~\ref{thm:master_precision}. For completeness, we give here the
case where the lattice is initially given by a generating family of
$p$ vectors, has rank $d$ and lives in an ambient space of dimension
$D$.
\begin{theorem}
  \label{thm:master_approx}
  Let $(\delta,\eta)$ be such that $1/4<\delta<1$ and
  $1/2<\eta<\sqrt{\delta}$. Let
  $c>\log\frac{(1+\eta)^2}{\delta-\eta^2}$.  Assume that we are given
  as input $(\Lambda, \inner{\cdot}{\cdot})$ a rank-$d$ lattice
  $(\Gram,L)$ described by $p\geq d$ generating vectors in a ambient
  space of dimension $D\geq d$. Further assume that it is
  approximately represented at accurary $N$ by the pair
  $(\inter{G_N},L)$ and  let $B$ denote the maximum entry in
  absolute value in $L^T\Gram L$.  Let $\ell=cd+\littleO{d}$ and
  $$ N \geq \ell + \log\left(B/\lambda_D(\Gram)\right)+\log(d).$$

  Then, the \TLL~of Figure~\ref{alg:LL_tilde} outputs
  a $(\delta,\eta)$-{\LLL}-reduced basis
  in time
  $$\bigO{DN\left(d^2N+p(p-d)\right)\Mult{d}}.$$
  Furthermore, if $\tau$
  denotes the number of main loop iterations, the running time is
  $\bigO{DN\left(dN+\tau\right)\Mult{d}}.$
\end{theorem}

\subsection{$L^2$ reduction with adaptive precision and accuracy.}

\subsubsection{Adaptive precision.}
\label{sec:adaptive_precision}

Since by construction the \TLL~Algorithm can detect that the choice
for internal precision $\ell$ is insufficient to correctly reduce the
lattice $\Lat$. The procedure can be wrapped in a loop that
geometrically increases precision $\ell$ after each unsuccessful
iteration. This yields an \emph{adaptive precision} reduction
algorithm~\ALLL.  Since the complexity of floating-point
multiplication is superlinear, the use of a geometric precision growth
guarantees that the total complexity of this lattice reduction is
asymptotically dominated by its final iteration.\footnote{In practice,
  for lattices of rank few hundreds it appears nonetheless that the
  computational cost of the previous iterations lies between $20\%$ and
  $40\%$ of the total cost.}

Moreover, the cost of operations in the floating-point realization of
interval arithmetic is at most four times the cost of floating-point
arithmetic at the same precision. Depending on the internal
representation used, this constant can even be improved.
As a consequence, for lattices that can be reduced with a low-enough
precision, it can be faster to use interval arithmetic than
floating-point arithmetic with the precision required by the bound
from Section~\ref{LLprecision}.

\subsubsection{Adaptive accuracy.}
We now turn to the setting of Section~\ref{sec:computational_setting},
where an algorithm or oracle $\mathcal{O}_\gamma$ can output
an integral representation of the
Gram matrix $\Gram_\gamma =
\left(\inner{\gamma_i}{\gamma_j}\right)_{(i,j)\in [1\,\cdots\, r]^2}$ at
arbitrary accuracy $n$. In that context, we need to determine both the
necessary accuracy and internal precision. When running
Algorithm~\ref{alg:LL_tilde} with some given accuracy and precision,
three outcomes are possible:
\begin{itemize}
  \item Either the reduction terminates in which case the lattice is
  \LLL-reduced, which implies that both accuracy and precision are
  sufficient.
  \item The Lovász condition fails to be tested correctly, which
  indicates an insufficient precision. In that case, we need to test
  whether the precision is lower than theoretical bound
  $T(d,\delta, \eta)$ given after Theorem~\ref{thm:master_precision}
  or not. In the latter case, we know that the accuracy needs to be
  increased.
  \item The algorithm detects a non-zero vector whose norm is given by
  an interval containing 0. This directly indicates insufficient
  accuracy.
\end{itemize}

Depending on the result of Algorithm~\ref{alg:LL_tilde}, we increase
the precision or the accuracy and restart. The corresponding
pseudo-code is given in Algorithm~\ref{alg:AdaptLLL}. Since the
precision and accuracy both follow a geometric growth, the
computation is dominated by its final iteration. In particular, we may
use the complexity bound given by Theorem~\ref{thm:master_approx}.

Note that when we increase the accuracy in
Algorithm~\ref{alg:AdaptLLL}, we also reset the precision to its
minimal value. This is a matter of preference that doesn't affect the
asymptotic complexity. In practice, it seems to be preferable.

It is important to note that we do need to precompute the eigenvalues
of the Gram matrix, since  Algorithm~\ref{alg:AdaptLLL} automatically
detects the needed accuracy.

\begin{algorithm}
  \BlankLine
  \Parameter{ $\delta \in (1/4,1), \eta \in (1/2, \sqrt{\delta})$,
    $\ell_0\in\NN$ initial precision
of the algorithm for floating-point representation, $n_0$ initial accuracy for
representing the scalar product, $g>1$ geometric growth factor.}
  \KwIn{$\gamma$ a basis of a lattice $(\Gamma, \inner{\cdot}{\cdot})$, and
  $\mathcal{O}_\gamma(n)$ an oracle that compute the integral representation of
the inner product $\inner{\cdot}{\cdot}$ at accuracy $n$.}
  \KwIn{A generating family represented by $L$ in $\gamma$ of a sublattice
  $\Lat\subset\Gamma$.}
  \KwResult{A $(\delta,\eta)$ \LLL-reduced basis of $\Lat$ represented as
  $L'\in\ZZ^{\rk(\Lat)\times \rk(\Lat)}$.}
  \BlankLine
  \tcp{Set initial values for accuracy and precision}
  \tcp{T$(d,\delta,\eta)$ is the theoretical bound given after Theorem~\ref{thm:master_precision}}
  $\ell \gets \ell_0$\; $n \gets n_0$\;
  $G \gets \mathcal{O}_\gamma(n)$\;
  succeed $\gets$ \False\;
  \Repeat{succeed = \True}{
    retcode $\gets$ \TLL$(G, L)$\;
    \lIf{retcode$=$ErrorNonPosDefinite}{\Return ErrorNonPosDefinite}

    \lIf{retcode$=$\textsc{OK}}{succeed $\gets$ \True}
      \ElseIf{retcode$=$ErrorPrecision}{
        $\ell' \gets \ell$\;
        $\ell \gets \min\left(\lceil g\,\ell\rceil,T(d,\delta,\eta), n\right)$\;
        \lIf{$\ell'=\ell$}{retcode $\gets$ ErrorAccuracy}
      }
      \If{retcode$=$ErrorAccuracy} {
        $\ell \gets \ell_0$\;
        $n \gets \lceil g\,n\rceil$\;
        $G \gets \mathcal{O}_\gamma(n)$\;
      }
    }
    \Return $L$
    \caption{The \ALLL algorithm.}
    \label{alg:AdaptLLL}
\end{algorithm}

\subsection{Possible generalizations}
The adaptative strategy we describe for \LLL~can be generalized to
other lattice reduction algorithm. In particular, enumeration
algorithms are possible within our framework, which allows the
implementation of the BKZ algorithm of~\cite{Schnorr87}.

It would be interesting to study a generalization to sieving
techniques to adapt them to approximate lattices.

\section{Application to Algebraic Number Theory}
We now present a direct application of our lattice reduction strategy
in algorithmic number theory. Namely, we consider some interesting
lattices sitting inside number fields: {\it ideal lattices.}

\subsection{Number fields, integers and ideal lattices}
\subsubsection*{Number fields}
A number field $\KK$ is a finite-dimensional algebraic extension of
$\QQ$. It can be described as:
$$\KK \cong \faktor{\QQ[X]}{(P)}=\QQ(\alpha),$$
where $P$ is a monic irreducible polynomial of degree $d$ in $\ZZ[X]$
and where $\alpha$ denotes the image of $X$ in the quotient.

Let $\left(\alpha_1,\dotsc,\alpha_d\right) \in \CC^d$ denote the
distinct complex roots of $P$. Then, there are $d$ distinct ring-embeddings
of $\KK$ in $\CC$. We define the $i$-th embedding $\sigma_i:\KK \to
\CC$ as the field homomorphism sending $\alpha$ to $\alpha_i$.

It is classical to distinguish embeddings induced by real roots,
a.k.a., {\it real embeddings} from embeddings coming from
(pairs of conjugate) complex roots, called {\it complex embeddings}.
Those arising from complex roots called complex embeddings.

Assume that $P$ has $r_1$ real roots and $r_2$ pairs of conjugate
complex roots, with $d=r_1+2r_2$. Since the embeddings corresponding
to conjugate roots are related by conjugation on $\CC$, we can either
keep a single complex root in each pair or replace each pair by the
real and imaginary part of the chosen root. This leads to the
\emph{Archimedean embedding} $\sigma$ defined as:
\[
\begin{array}{cccl}
  \sigma:& \KK & \longrightarrow & \RR^d\\
  &x & \longmapsto & \left(\sigma_1(x),\ldots,\sigma_{r_1}(x),
  \sqrt{2}\mathfrak{R}(\sigma_{r_1+1}(x)), 
  \sqrt{2}\mathfrak{I}(\sigma_{r_1+1}(x)), \ldots
  \right)^T
\end{array}
\]

This embedding allows us to define a real symmetric bilinear form on $\KK$:
 \[\inner{\vect{a}}{b}_\sigma=\sigma(a)\cdot\sigma(b)=\sum_{i=1}^{d}
 \sigma_i(\vect{a})\overline{\sigma_i(b)}.\]
The second equality explains the presence of the normalization factors
$\sqrt{2}$ in the definition of $\sigma$. Note that the form is
positive definite, thus endowing $\KK$ with an Euclidean structure.

\subsubsection*{Integers}
Any element $\gamma$ of $\KK$ has a minimal polynomial, defined as
the unique monic polynomial of least degree among all polynomials of
$\QQ[X]$ vanishing at $\gamma$.  The algebraic number $\gamma$ is said
to be \emph{integral} if its minimal polynomial lies in $\ZZ[X]$. The
set of all integers in $\KK$ forms a ring, called the {\it ring of
integers} of $\KK$ and denoted $\order_\KK$. It is also a free
$\ZZ$-module of rank $d$. A basis $(w_1,\ldots, w_d)$ of $\order_\KK$
(as a $\ZZ$-module) is called an integral basis of $\KK$.

As a consequence, using the bilinear form
$\inner{\cdot}{\cdot}_{\sigma}$, we can view $\order_\KK$ as a
lattice.

\subsubsection*{Ideals}
An ideal of $\order_\KK$ is defined as an $\order_\KK$-submodule of
$\order_\KK$. In particular, it is a $\ZZ$-submodule of rank~$d$.
Every ideal $I$ can be described by a two-element
representation, i.e. expressed as $I=\alpha \order_\KK
+ \beta \order_\KK,$ with $\alpha$ and $\beta$ in
$\order_\KK$. Alternatively, every ideal can also be described by a
$\ZZ$-basis formed of $d$ elements.

\subsection{Lattice reduction for ideals}
With the above notations, we can directly use our lattice reduction
algorithm to reduce an ideal lattice. More precisely, given an
integral basis $(w_1,\ldots, w_d)$ and a two-element representation of
$I$ by $\alpha$ and $\beta$, we proceed as follows:
\begin{enumerate}
  \item Define the Gram matrix $\Gram_{w}$ with entries
    $\inner{w_i}{w_j}_{\sigma}$. It can be computed to any desired
    precision from approximations of the roots of $P$. The roots
    themselves can be computed, using, for example, the
    Gourdon-Sch\"onhage algorithm~\cite{Gourdon96}.
  \item Let $L$ be the matrix formed of the (integral) coordinates of
    $(\alpha w_1, \ldots, \alpha w_d)$ and $(\beta w_1, \ldots, \beta
    w_d)$ in the basis  $(w_1,\ldots, w_d).$
  \item Directly apply Algorithm~\ref{alg:AdaptLLL} to $(\Gram_w,L).$
\end{enumerate}

The same thing can be done, {\it mutantis mudantis}, for an ideal
described by a $\ZZ$-basis.

\subsubsection*{A well-known special case} For some number fields, the
Gram matrix is $\Gram_w$ is integral. In that case, the use of
Algorithm~\ref{alg:AdaptLLL} isn't necessary and one can directly work
with an exact lattice. This is described for the special case of
reducing the full lattice corresponding to the ring of integers
in~\cite[Section 4.2]{Belabas04} for totally real fields. It can be
generalized to CM-fields, since they satisfy the same essential
property of having an integral Gram matrix. The same application is
also discussed in~\cite[Section 4.4.2]{Cohen93}.

\subsubsection*{Non integral case}
For the general case where the Gram matrix is real, \cite{Belabas04}
propose to multiply by $2^e$ and round to the closest integer. It also
gives a bound on the necessary accuracy $e$  as the logarithm of (the
inverse of) the smallest diagonal entry in the Cholesky
decomposition of the Gram matrix. In some sense, this is similar to
our approach. However, without any auxiliary information on this
coefficient, it is proposed to continue {\it increasing $e$
as long as it is deemed unsatisfactory.}

By contrast, termination of our algorithm guarantees that lattice
reduction is completed and that the output basis is \LLL-reduced.

\providecommand{\href}[2]{#2}
\providecommand{\arxiv}[1]{\href{http://arxiv.org/abs/#1}{arXiv:#1}}
\providecommand{\url}[1]{\texttt{#1}}
\providecommand{\urlprefix}{URL }

\medskip

\appendix
\section{Proof of Proposition~\ref{prop:norm_bound}}
We now show the more general statement for a
$(\delta,\eta)$-\LLL~reduced basis $(b_1, \ldots, b_d)$ of $(\Lat,
\inner{\cdot}{\cdot})$. Namely that for any $1\leq k\leq d$ we have:
  \[
    \Vol{b_1,\ldots,b_k}\leq \left(\delta-\eta^2\right)^{-\frac{(d-k)k}{4}}
    \Vol{\Lat}^{\frac{k}{d}}.
    \]

\begin{proof}
  Using the Lovász condition at index $1\leq i<d$, we write:
  \[
    \delta \|\pi_i(b_i)\|^2 \leq \|\pi_i(b_{i+1}) \|^2 = \|\pi_{i+1}(b_{i+1}) \|^2 +
    \mu_{i,i+1}^2\|\pi_i(b_i)\|^2
  \]
  Thanks to the size-reduction condition, this implies:
  \begin{equation}
    \label{eq:single_majoration_LLL}
    \forall i\in\{1,\ldots,d-1\}, \quad \|\pi_i(b_i)\|^2 \leq
    \left(\delta-\eta^2\right)^{-1}  \|\pi_{i+1}(b_{i+1})\|^2.
  \end{equation}

  Let $K$ denote $\left(\delta-\eta^2\right)^{-1/2}$
  and $\ell_i$ be the norm of the vector $\pi_i(b_i)$. Then,
  Equation~\eqref{eq:single_majoration_LLL} becomes:
  \[
    \forall i\in\{1,\ldots,d-1\}, \quad \ell_i \leq K \ell_{i+1}.
  \]

  Recall that $\Vol{b_1, \ldots, b_k} = \prod_{i=1}^k \ell_i$. This
  implies that for any $j>k$:
  \[
  \Vol{b_1, \ldots, b_k} \leq \prod_{i=1}^k K^{j-i}\ell_j=K^{k(2j-k-1)/2}\cdot\ell_j^k.
  \]
  Thus:
  \[
    \begin{aligned}
      \Vol{b_1, \ldots, b_k}^{d} = \left(\prod_{i=1}^k \ell_i\right)^d
      &\leq \left(\prod_{i=1}^k \ell_i\right)^k
      \, \prod_{j=k+1}^d K^{k(2j-k-1)/2}\cdot\ell_j^k\\
      &\leq \left(\prod_{i=1}^d \ell_i\right)^k
      \, K^{\sum_{j=k+1}^d{k(2j-k-1)/2}}\\
      &\leq \Vol{\Lat}^k K^{\frac{d(d-k)k}{2}}.
    \end{aligned}
  \]
\end{proof}

\end{document}